\documentclass[letterpaper]{article} 
\usepackage{aaai2027}  
\usepackage[hyphens]{url}  
\usepackage{graphicx} 
\usepackage{natbib}  
\usepackage{caption} 
\usepackage{algorithm}
\usepackage{algorithmic}
\usepackage{newfloat}
\usepackage{listings}
\DeclareCaptionStyle{ruled}{labelfont=normalfont,labelsep=colon,strut=off} 
\floatstyle{ruled}
\newfloat{listing}{tb}{lst}{}
\floatname{listing}{Listing}

\usepackage{booktabs}
\usepackage{makecell}
\usepackage{tabularx}
\usepackage{amsfonts}
\usepackage{nicefrac}
\usepackage{microtype}
\usepackage{xcolor}
\usepackage{amssymb}
\usepackage{amsmath}
\usepackage{amsthm}
\usepackage{multirow}
\usepackage{subcaption}

\newtheorem{theorem}{Theorem}

\newtheorem{definition}{Definition}

\newtheorem{proposition}{Proposition}

\title{STITCH-RAG: Spatio-Temporal Influence Tracing over Topic Hypergraphs for Multi-Hop Retrieval-Augmented Generation}
\author{
	Haodong Yang,
	Mengzhu Chen,
	Jia Cai\thanks{Corresponding author}
}

\affiliations{
	School of Statistics And Data Science, Guangdong  University of Finance $\&$ Economics\\
	Guangzhou, China\\
	haodongyang@student.gdufe.edu.cn, mengzhuchen@student.gdufe.edu.cn, jiacai1999@gdufe.edu.cn
}

\begin{document}

\maketitle

\begin{abstract}
Multi-hop retrieval-augmented generation requires a retriever to connect evidence distributed across documents while preserving a concise, faithful generation context. Existing indexes leave two complementary gaps: chunk-based RAG can break cross-passage evidence chains, whereas an unlabeled pairwise projection without generating-topic provenance cannot jointly preserve topic-level co-participation and per-occurrence entity descriptions. We propose STITCH-RAG, a hypergraph-based framework with three coupled components. First, a semi-merged topic hypergraph encodes multi-entity co-participation as topic-summary hyperedges while retaining per-chunk entity states linked by canonical-name equivalence. Second, spatio-temporal influence bridging propagation (STIBP) combines topic-space propagation with deterministic chunk-index linkage across name-equivalent states under frequency-adaptive decay. Third, continuous STIBP scores replace binary entity-match seeds in localized Personalized PageRank (PPR). We characterize the condition under which this prior assigns more PPR mass to ground-truth evidence than a binary prior. Under the reported protocol, STITCH-RAG attains the highest reported Contain-Acc and LLM-Acc point estimates among the compared methods on HotpotQA and 2WikiMultiHopQA, and higher Recall@8 than the methods included in the standardized retrieval comparison. Results on the mixed-domain benchmark remain auxiliary preference-based evidence because only LLM-judged accuracy is available.
\end{abstract}

\section{Introduction}

Standard dense-retrieval RAG partitions documents into fixed-size chunks and ranks them by query similarity. This representation can miss multi-hop chains that cross chunk or document boundaries. Structured RAG adds entity and relation indexes, but the published representations considered here do not jointly retain two forms of context. First, an unlabeled pairwise projection without a generating-topic identifier discards the topic that produced a multi-entity co-occurrence. This limitation does not apply to pairwise graphs augmented with relation labels, provenance, event identifiers, or equivalent higher-order attributes. Second, fully merging repeated entity mentions can erase chunk-specific roles, whereas treating every mention independently removes an explicit cross-chunk identity link. Both choices introduce retrieval noise when evidence depends on topic context and local entity state.

\begin{figure}[t]
\centering
\includegraphics[width=0.98\linewidth]{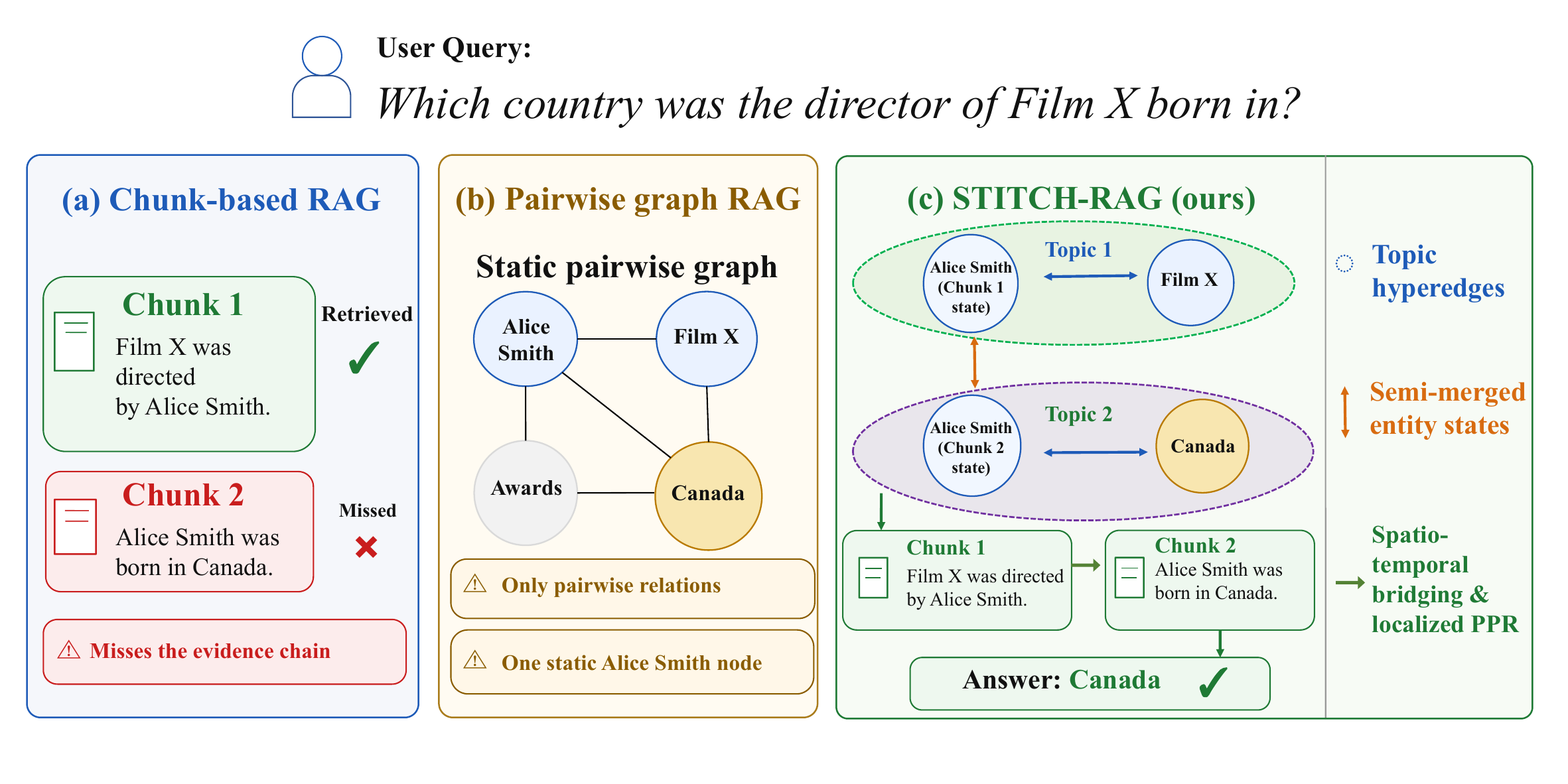}
\caption{Representation trade-offs for an example query about
the birthplace of Film X's director. (a) Chunk-based RAG can return
locally similar but disconnected passages. (b) The
illustrated unlabeled pairwise/full-merge representation retains the
path through Alice Smith, but its single static entity node mixes the
evidence-bearing birthplace context with unrelated contexts such as
awards. The panel depicts one unlabeled full-merge
representation, not all pairwise graphs. (c) STITCH-RAG preserves the two local
Alice Smith states and links them by name equivalence, allowing
query-conditioned propagation through topic hyperedges and across
chunk-specific entity states.}
\label{main:example}
\end{figure}

Figure~\ref{main:example} illustrates these representation trade-offs. In the controlled unlabeled pairwise/full-merge example, the graph retains a route from Film X to Canada, but a single Alice Smith node cannot distinguish the birthplace-bearing context from the award context. STITCH-RAG instead represents topic-level co-participation with hyperedges and retains chunk-specific entity states as separate nodes. An LLM extracts topic summaries as hyperedges and context-aware entity descriptions as nodes. The semi-merged hypergraph preserves those descriptions and links mentions with identical normalized canonical names. This link is a name-equivalence heuristic, not a ground-truth entity linker. Throughout the paper, \textit{spatial} denotes topical co-participation within a hyperedge, whereas \textit{temporal} denotes proximity in deterministic preprocessing indices. Across documents, this quantity is an index-proximity heuristic rather than a reading order, chronology, event time, or causal relation.

At retrieval time, dense similarity activates entity nodes. Spatial bridging transfers influence to co-participants in the same topic hyperedge, whereas index-proximity bridging transfers influence across name-equivalent states through bounded frequency-adaptive decay and query-relevance gating. The resulting continuous scores are projected onto chunks and used as priors for localized approximate Personalized PageRank (PPR), replacing binary entity-match initialization with graded, context-aware diffusion. On HotpotQA, 2WikiMultiHopQA, and Mix, STITCH-RAG attains the highest reported Contain-Acc and LLM-Acc point estimates under the stated protocol. Recall@8 is compared only where standardized retrieval outputs are available. LLM-judged metrics remain auxiliary (Section~\ref{sec:metrics}).

The novelty claim concerns the coupling of
topic-preserving hyperedges, context-specific entity states,
operational name-equivalence links, and continuous PPR
initialization in one retrieval pipeline. No single component is
claimed to be new in isolation. The main contributions are as
follows:

\begin{itemize}

\item \textbf{A semi-merged hypergraph with a
controlled recoverability characterization.} The construction
preserves per-chunk entity descriptions as distinct nodes while
exposing canonical-name equivalence through an
operational name-equivalence map.
Proposition~\ref{prop:semimerge_recoverability} shows that, among
the three controlled representations considered here,
semi-merging alone retains both properties required by the stated
state-specific propagation rule. The controlled
merge diagnostic reports a higher LLM-Acc point estimate for the
implemented semi-merge variant than for the two implemented
alternatives. Because retrieval metrics, structural compression,
and run-level variances are unavailable for these variants, this is
supporting, not conclusive, evidence.

\item \textbf{Topic-space and index-proximity
influence bridging with frequency-adaptive decay.} STIBP operates
before PPR through topic-hyperedge propagation and deterministic
chunk-index linkage across name-equivalent entity states. The
bounded decay $f(\Delta t)=\exp(-\tanh(\alpha\Delta t))$ uses
$\alpha=n_u/\bar n$ to attenuate index-distant links more strongly
for frequent names. Theorem~\ref{thm:stibp_snr} provides a
one-sided bound for contributions propagated from one
activated source. It neither models chronology nor characterizes
the final scores after max aggregation over multiple sources.

\item \textbf{Continuous influence priors with a
conditional evidence-mass characterization.} Existing PPR-based
methods~\cite{zhuang2025linearrag,gutierrez2024hipporag} commonly
use binary entity-match initialization. STITCH-RAG instead derives
a non-negative, normalized chunk prior from graded STIBP scores.
Proposition~\ref{thm:ppr_evidence_mass} gives an exact PPR
evidence-mass identity and a sufficient alignment condition under
which this prior assigns more evidence mass than a binary prior.
The strict binary-only diagnostic in
Table~\ref{tab:ablation_extended} is not used to verify this
condition because binary seed coverage, empty-match behavior, and
prior-support size are not measured separately.

\end{itemize}
The propositions and theorem characterize the chosen
representation, decay, and PPR prior. They clarify the design assumptions but do not
constitute end-to-end superiority guarantees.

\section{Related Work}

Structured RAG methods differ primarily in index representation and query-signal propagation. Under their published representations, the methods compared here do not jointly retain topic-level co-participation, per-occurrence descriptions, and operational name-equivalence links. STITCH-RAG couples these properties to construct graded localized-PPR priors. It does not claim novelty for each component in isolation.

\noindent\textbf{Graph-based RAG.}
GraphRAG, LightRAG, NodeRAG, and LinearRAG instantiate entity--relation, dual-layer, heterogeneous, and name-matched PPR indexes, respectively~\cite{edge2024local,guo2024lightrag,xu2025noderag,zhuang2025linearrag}. For an unlabeled pairwise projection, the generating-topic identity is unavailable at retrieval time (Proposition~\ref{prop:topic_nonidentifiability}). This scope excludes provenance-enriched pairwise graphs. SubGraphRAG and GFM-RAG~\cite{li2025simple,luo2025gfmrag} prioritize local precision, whereas HippoRAG variants~\cite{gutierrez2024hipporag,gutierrez2025rag} apply global PPR with binary matching. STITCH-RAG preserves topic identity and linked local states, then uses their graded influence as the PPR prior (Proposition~\ref{thm:ppr_evidence_mass}).

\noindent\textbf{Hypergraph-based RAG.}
HyperGraphRAG, Cog-RAG, and CogniRAG use relational, thematic, and causal hyperedges~\cite{luo2025hypergraphrag,hu2026cog,he2026cognirag}. In the controlled merge constructions, full merging removes local states and no merging removes $\phi$, whereas semi-merging retains both $\phi$ and $\psi$. Hyper-RAG, Cog-RAG, and IGMiRAG trade traversal, filtering, or routing cost against coverage~\cite{feng2026hyperrag,hou2026igmirag}. STITCH-RAG instead propagates only within the activated subgraph. Unlike the trained precedence model in OKH-RAG~\cite{wu2026knowledge}, its index-proximity decay is training-free and targets complex multi-hop queries, where structured retrieval is most useful~\cite{xiang2025whentousegraphs}.

\noindent\textbf{Complementary directions.}
Adaptive routing, long-context retrieval, and granularity selection address complementary design choices~\cite{yan2024corrective,jeong2024adaptive,asai2024selfrag,li2024structrag,li2024retrieval,xu2024retrieval,chen2024dense,sarthi2024raptor,kim2025autorag}. STITCH-RAG targets boundary fragmentation in large or dynamic knowledge bases, where one-time indexing can be amortized across repeated queries~\cite{ovadia2024finetuning,balaguer2024rag}. Appendix~\ref{app:extended_related} provides detailed comparisons.

\section{Method}
\subsection{Overview and Design Rationale}
\label{sec:rationale}

STITCH-RAG implements a \emph{compress-then-amplify} retrieval pipeline (Figure~\ref{main:framework}). During compression, documents are transformed into a semi-merged topic hypergraph $\mathcal{H}=(\mathcal{V},\mathcal{E},\phi,\psi)$ whose hyperedges are topic summaries and whose nodes are context-aware entity descriptions for individual chunks. During amplification, spatio-temporal influence bridging propagation (STIBP) identifies relevant entity-state pathways, and localized approximate Personalized PageRank (PPR) diffuses the resulting continuous scores over a chunk-level graph.

\begin{figure*}[t]
\centering
\includegraphics[width=0.98\textwidth]{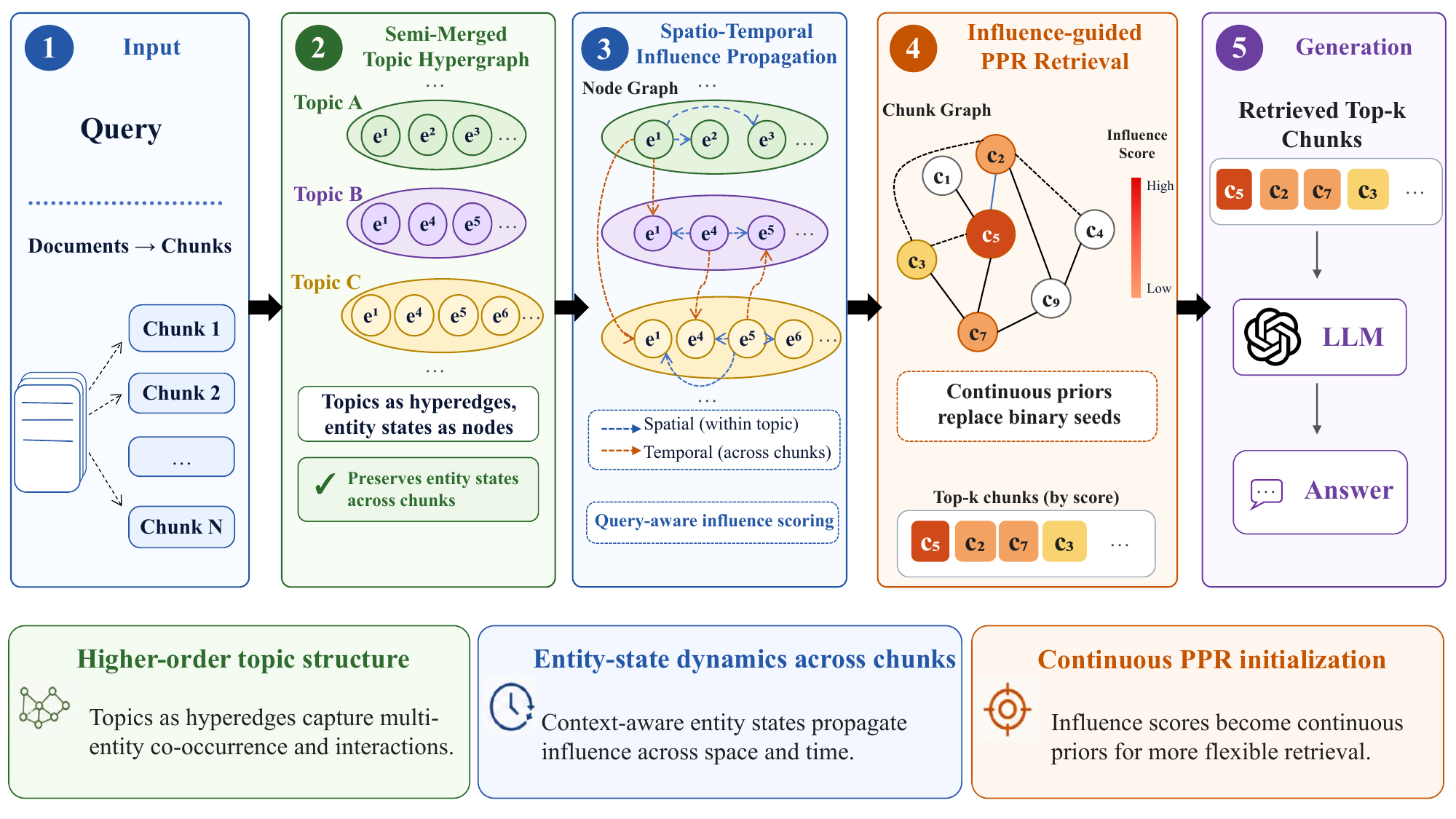}
\caption{Framework of STITCH-RAG:
(1) document chunking; (2) semi-merged topic hypergraph
construction, with topic-summary hyperedges and
entity-description nodes linked by $\phi$; (3) STIBP
through spatial and temporal channels;
(4) influence-guided chunk scoring as a continuous PPR
teleportation prior; and (5) top-$k$ chunk retrieval for
answer generation.}
\label{main:framework}
\end{figure*}

Semi-merging is required because temporal bridging uses $\phi$, whereas query-conditioned activation requires local descriptions $\psi$. Proposition~\ref{prop:semimerge_recoverability} formalizes this requirement for the controlled constructions. STIBP scores entity states on the hypergraph, and PPR diffuses those scores over the induced chunk graph. Table~\ref{tab:ablation_extended} tests their roles through controlled replacements.

\subsection{Hypergraph Construction Phase}

Given document collection $D$, we produce chunks
$C=\{c_k\}_{k=1}^{N_C}$ with at most $c_{\max}$ tokens each.
An LLM with prompt $p_{ext}$ processes each $c_k$ to extract
tuples $g_{k,r}=(e_{k,r},V_{e_{k,r}})$, where $e_{k,r}$ is
a topic-summary hyperedge and $V_{e_{k,r}}$ is its entity set:
\[
\mathcal{G}=\{g_{k,r}\mid g_{k,r}\in \operatorname{LLM}(c_k\mid p_{ext}),\; c_k\in C\}.
\]
Each node $v=(v^{name},v^{des})$ stores a canonical name and a chunk-specific description. The heuristic $\phi:\mathcal{V}\to\Sigma$ groups identical normalized names, while $\psi:\mathcal{V}\to\mathcal{D}$ records local context; members of each group are ordered by deterministic preprocessing index. Chunks remain the retrieval targets. A topic hyperedge with $m$ entities uses $m$ incidence links rather than $\binom{m}{2}$ pairwise links and retains its generating-topic identity (Proposition~\ref{prop:topic_nonidentifiability}).

\subsubsection{Structural Advantage over Pairwise Projection}

For $\mathcal{H}=(\mathcal{V},\mathcal{E})$, the projection $\Pi(\mathcal H)$ connects $u\ne v$ whenever $u,v\in e$ for some $e$. It discards the identity of the generating hyperedge. Proposition~\ref{prop:topic_nonidentifiability} therefore concerns retrieval from this unlabeled, provenance-free input, not pairwise graphs whose attributes recover the grouping (Appendix~\ref{app:SecD}).

\begin{proposition}[Non-identifiability of Pairwise Topic Projection]
\label{prop:topic_nonidentifiability}
There exist two topic hypergraphs $\mathcal{H}_1$ and
$\mathcal{H}_2$ such that $\Pi(\mathcal{H}_1)=\Pi(\mathcal{H}_2)$,
but query-conditioned spatial propagation over $\mathcal{H}_1$
and $\mathcal{H}_2$ assigns different influence scores to the
same target entities. Therefore, any retrieval rule
whose input is limited to $\Pi(\mathcal H)$, without labels,
provenance, or attributes that recover the generating hyperedges,
cannot in general reproduce topic-selective propagation over those
hyperedges.
\end{proposition}

\subsubsection{Recoverability Benefit of Semi-Merging}

The controlled constructions differ only in their exposure of $\phi$ and retention of $\psi$ (Definitions~\ref{def:semi_merge}--\ref{def:full_no_merge}). For the propagation rule below, only semi-merging retains both; the result is a controlled recoverability comparison rather than a universal claim.

\begin{proposition}[Recoverability Separation of Semi-Merged Entity States]
\label{prop:semimerge_recoverability}
Let $q$ be a query and let $G=\{v_1,\ldots,v_m\}$ be a
name-equivalence group with $\phi(v_i)=\phi(v_j)$ for all $i,j$.
Suppose $v_s\in G$ satisfies $A_q(v_s)>\delta$, and $v_r\in G$
lies in a ground-truth evidence chunk with $A_q(v_r)\leq\delta$.
Assume $v_r$ is not reachable from $v_s$ through any activated
topic hyperedge, so that its only structural path from $v_s$
is through name-equivalence temporal linkage. If
$S_q(e_{v_r})>0$ and $\operatorname{sim}(q,c_{v_r})>0$, then
$v_r$ receives no temporal influence under
the no-merge construction $\mathcal{H}_N$, which treats
each mention as an independent node with no cross-chunk
$\phi$; cannot be assigned state-specific influence under
the full-merge construction $\mathcal{H}_F$, which
collapses each name-equivalence group into a single node; and
receives strictly positive temporal influence under the
semi-merged $\mathcal{H}$.
\end{proposition}

\subsection{Retrieval Phase}

Given query $q$, retrieval has four stages: (1) dense-similarity activation of entity nodes and topic hyperedges; (2) STIBP over spatial and temporal channels to obtain graded entity-level influence; (3) projection of entity influence onto chunks with name-equivalence group-size weighting; and (4) localized PPR on the $\phi$-induced chunk graph, using the projected influence as a continuous teleportation prior.

\subsubsection{Initial Node and Hyperedge Activation}

We embed the query and entity descriptions with the same model and assign
\[
A_q(v)=
\begin{cases}
\operatorname{sim}(q,v^{des}), & \operatorname{sim}(q,v^{des})>\delta,\\
0, & \text{otherwise},
\end{cases}
\]
with $\delta=0.5$. Hyperedges use the unthresholded
$S_q(e)=\operatorname{sim}(q,e)$ because spatial gating already
attenuates low-relevance topics.

\subsubsection{Spatio-Temporal Influence Bridging Propagation}

STIBP maintains $A_q(v)$, $a^{space}(v)$, and $a^{time}(v)$ per node;
within each propagated channel, it retains the maximum source value.

\noindent\textbf{Spatial bridging.}
For activated $u\in e$ and each $v\in V_e\setminus\{u\}$:
\[
a^{space}(v)=A_q(u)\cdot S_q(e).
\]
This update requires both source activation and topic relevance.

\noindent\textbf{Temporal (index-proximity) bridging.}
For the name-equivalence group
$V_u^*=\{v\in\mathcal{V}\mid\phi(v)=\phi(u)\}$ of activated
$u$, set $\alpha=n_u/\bar{n}$ where $n_u=|V_u^*|$ and
$\bar{n}$ is the corpus-wide average group size.
Here $t_v$ denotes the deterministic preprocessing index; across documents, the resulting distance is order-dependent index proximity rather than chronology. For each
$v\in V_u^*\setminus\{u\}$ with $\Delta t=|t_u-t_v|$, define
the contribution proposed by source $u$ as
\[
a^{time}_u(v)=A_q(u)\cdot S_q(e_v)\cdot
\exp(-\tanh(\alpha\Delta t))\cdot \operatorname{sim}(q,c_v).
\]
Here $e_v$ denotes a hyperedge containing $v$. When $v$
belongs to several hyperedges, $S_q(e_v)$ stands for
$\max_{e\ni v}S_q(e)$, matching the max-aggregation rule
above. With $\mathcal A_q=\{u:A_q(u)>0\}$, the stored temporal
channel is
\[
a^{time}(v)=\max\!\left(0,
\max_{\substack{u\in\mathcal A_q:\,u\ne v,\\
\phi(u)=\phi(v)}}a^{time}_u(v)\right).
\]
The decay function $f(\Delta t)=\exp(-\tanh(\alpha\Delta t))$ lies in $(e^{-1},1]$: it retains nonzero mass for distant evidence while attenuating it. Because $\alpha\propto n_u$, distant states of frequent names receive stronger attenuation (Table~\ref{tab:decay_comparison}, Appendix~\ref{app:additional_results}).

The final entity influence aggregates all three channels:
\[
a(v)=\operatorname{mean}\bigl(A_q(v),\,a^{space}(v),\,a^{time}(v)\bigr).
\]
Mean aggregation weights the three channels equally and induces the same ranking as summation because every node has all three channels.

\subsubsection{Per-Source Analysis of
Frequency-adaptive STIBP}

The following per-source result characterizes bounded decay under explicit signal--noise assumptions. It is not an end-to-end retrieval guarantee.

\begin{theorem}[Per-Source Lower-Bound Separation under Frequency-Adaptive Temporal Decay]
\label{thm:stibp_snr}
Let $u$ be an activated entity state with $A_q(u)=a>0$ and
name-equivalence group $V_u^*$. For each $v\in V_u^*\setminus\{u\}$,
let $\Delta_v=|t_u-t_v|$ and
$g_v=S_q(e_v)\operatorname{sim}(q,c_v)$. Let $\mathcal{R}$ and
$\mathcal{N}$ be nonempty disjoint subsets
of $V_u^*\setminus\{u\}$ satisfying
$g_r\geq\beta$ for all $r\in\mathcal{R}$,
$0<g_n\leq\epsilon$ and $\Delta_n\geq\tau$ for all
$n\in\mathcal{N}$, with $\beta>\epsilon>0$ and $\tau>0$. With
$\alpha=n_u/\bar{n}$, where $n_u=|V_u^*|$:
\[
\frac{\max_{r\in\mathcal{R}}a^{time}_u(r)}
{\sum_{n\in\mathcal{N}}a^{time}_u(n)}
\geq
\frac{\beta}{|\mathcal{N}|\epsilon}\exp\bigl(\tanh(\alpha\tau)-1\bigr).
\]
For fixed $\mathcal R$, $\mathcal N$, $\beta$, $\epsilon$,
$\tau$, and $\bar n$, the right-hand side is non-decreasing
in $\alpha$.
\end{theorem}
The bound excludes final max aggregation and end-to-end retrieval. It also permits over-suppression of distant relevant states.

\subsubsection{Chunk Score Projection and Localized Approximate PPR}

STIBP scores are projected onto chunks and diffused on a $\phi$-derived chunk graph without further LLM calls. Let $C_q$ contain direct ANN candidates and chunks incident to nonzero-influence states; projection and PPR are restricted to this query-induced set.
Entity influence is projected onto chunks by
\[
{\rm Score}(c)=\sum_{v\in\mathcal{V}(c)} a(v)\cdot \ln(1+N_v),
\]
where $N_v=|\{w\in\mathcal{V}\mid\phi(w)=\phi(v)\}|$. We form the
non-negative prior score
\[
\widetilde I(c)=
\lambda\max\{0,\operatorname{sim}(q,c)\}
+\ln(1+{\rm Score}(c)),
\]
and normalize it over $C_q$ as
\[
p_{\mathrm{ST}}(c)=
\begin{cases}
\widetilde I(c)/\sum_{z\in C_q}\widetilde I(z),
& \text{if }\sum_z\widetilde I(z)>0,\\
1/|C_q|, & \text{otherwise}.
\end{cases}
\]
The uniform fallback handles an all-zero query; $\lambda=0.3$ is selected as described in Section~\ref{sec:hyper}. The chunk graph connects
chunks sharing a name-equivalent entity:
\[
\mathcal{B}(c)=\{c'\mid c'\neq c,\;\exists u\in\mathcal{V}(c),v\in\mathcal{V}(c'):
\phi(u)=\phi(v)\}.
\]
Let $P$ be the column-normalized transition matrix, with $P_{c'c'}=1$ for dangling chunks. With $d=0.85$, the PPR iteration is
\[
\pi^{(t)}=(1-d)p_{\mathrm{ST}}+dP\pi^{(t-1)},
\qquad \pi^{(0)}=p_{\mathrm{ST}},
\]
We iterate until $\max_c|\pi^{(t)}(c)-\pi^{(t-1)}(c)|<10^{-6}$ and pass the top-$k$ chunks ranked by $\pi^{(T)}$ to the generator. The following identity supports the design. Its sufficient alignment condition is not independently verified.

\begin{proposition}[Discounted Evidence-Mass Identity under Prior Alignment]
\label{thm:ppr_evidence_mass}
Let $p_{\mathrm{ST}}$ be the normalized chunk prior induced
by STIBP and let $p_{\mathrm{B}}$ be a normalized
binary entity-match prior. For a ground-truth evidence set $\mathcal{R}$, define
its discounted evidence-reachability footprint as
\[
h_{\mathcal{R}} =
(1-d)\sum_{\ell=0}^{\infty}d^\ell (P^\ell)^\top \mathbf{1}_{\mathcal{R}},
\]
where $P$ is the column-stochastic transition matrix of the
query-induced chunk graph. The evidence-mass
difference is exactly
\[
M_{\mathcal{R}}(p_{\mathrm{ST}})
-M_{\mathcal{R}}(p_{\mathrm{B}})
=h_{\mathcal{R}}^\top(p_{\mathrm{ST}}-p_{\mathrm{B}}).
\]
Consequently, if
$h_{\mathcal{R}}^\top(p_{\mathrm{ST}}-p_{\mathrm{B}})\geq \gamma>0$,
then
\[
M_{\mathcal{R}}(p_{\mathrm{ST}})
\geq
M_{\mathcal{R}}(p_{\mathrm{B}})+\gamma,
\]
where $M_{\mathcal{R}}(p)=\mathbf{1}_{\mathcal{R}}^\top\pi(p)$
denotes the PPR evidence mass assigned to $\mathcal{R}$ under
prior $p$.
\end{proposition}

Here $h_{\mathcal R}(c)$ is the discounted forward-walk reachability from $c$ to the evidence set. The binary diagnostic does not verify the alignment condition because reachability, seed coverage, empty matches, and support size are unmeasured (Appendix~\ref{app:SecD}).

Symbolic indexing requires $O(N_C+N_V+N_E+M+\textstyle\sum_g n_g\log n_g)$ time and linear storage. After candidate formation, STIBP and PPR scale with activated structures; Appendix~\ref{app:complexity} gives the full online-complexity expression.
\section{Experiments}

We evaluate STITCH-RAG along four axes: end-to-end answer accuracy (RQ1), answer quality beyond correctness (RQ2), module-level causality (RQ3), and computational efficiency (RQ4).

\subsection{Experimental Setup}

\subsubsection{Datasets}
We use three datasets covering structured multi-hop QA and open-ended domain-specific questions. For HotpotQA~\cite{yang2018hotpotqa} and 2WikiMultiHopQA~\cite{ho2020constructing}, we follow the HippoRAG~\cite{gutierrez2024hipporag} protocol and draw a stratified sample of 1{,}000 questions per dataset (composition in Appendix~\ref{app:implementation}). The sampled questions are fixed across methods. We report three repeated end-to-end runs, not three resampled datasets; at temperature zero, residual variation can arise from API and extraction-service nondeterminism. Mix contains 512 open-ended questions and serves as an exploratory cross-domain evaluation. Because Mix lacks supporting-fact annotations and reports only LLM-Acc, it is not an independent objective test of generalization.

\subsubsection{Baselines}
The baselines span the main RAG design choices: zero-shot generation and dense-retrieval RAG as lower bounds; LightRAG~\cite{guo2024lightrag} for relation extraction; LinearRAG~\cite{zhuang2025linearrag} as the closest structural predecessor, using PPR with binary entity-match initialization over a flat entity graph; HippoRAG~\cite{gutierrez2024hipporag} for PPR-based graph memory; and Cog-RAG~\cite{hu2026cog} and Hyper-RAG~\cite{feng2026hyperrag} as recent hypergraph-based methods.

\subsubsection{Metrics}\label{sec:metrics}
We report three end-to-end metrics: Contain-Acc, which tests whether the generated answer contains the reference substring; LLM-Acc, in which Qwen-Max judges answers generated by Qwen3.5-flash; and EM, in which a distilled answer must exactly match the reference. Because Qwen-Max and Qwen3.5-flash belong to the same model family, LLM-Acc can reflect same-family judge preference. Contain-Acc and EM do not use an LLM judge and therefore serve as controls. We treat Contain-Acc, EM, and Recall@8 as primary objective evidence on HotpotQA and 2Wiki, and treat LLM-Acc and pairwise quality dimensions as auxiliary preference-based evaluations. For retrieval, we report passage-level $\operatorname{Recall@8}(q)=|G_q\cap R_q^8|/|G_q|$ using supporting-fact annotations. This metric tests whether answer-quality differences are consistent with retrieval coverage, but does not isolate retrieval as the causal source. Mix lacks such annotations, so no retrieval metric is reported for it.

\subsubsection{Implementation Details}
All methods use \texttt{text-embedding-v4} with 1024 dimensions and Qwen3.5-flash at temperature zero. On a 200-question HotpotQA development subset, replacing the original embeddings of LightRAG and LinearRAG changes LLM-Acc by at most 0.4 and 0.2 points, respectively. This check limits embedding sensitivity for these two baselines on that subset, but does not establish the same property for every baseline. STITCH-RAG uses $k=8$, $\delta=0.5$, and $d=0.85$, selected on the same held-out set and fixed across datasets. All experiments run through the Alibaba Cloud API; every repeated run uses the same evaluation questions.

\subsection{Generation Accuracy (RQ1)}

\begin{table*}[t]
\centering
\caption{Answer accuracy on HotpotQA, 2Wiki, and Mix. Best result
per column bolded. Only STITCH-RAG was repeated
end to end; its entries average three runs on the same fixed
question samples, with all run-level standard deviations below
0.004. Baseline run-level variances are unavailable. The table
therefore ranks reported point estimates but does not establish
statistical superiority.}
\label{tab:q1}
\resizebox{0.95\linewidth}{!}{
\begin{tabular}{lccccccc}
\toprule
Method & \multicolumn{3}{c}{HotpotQA} &
\multicolumn{3}{c}{2Wiki} & Mix \\
\cmidrule(lr){2-4}\cmidrule(lr){5-7}\cmidrule(lr){8-8}
& Contain-Acc & LLM-Acc & EM &
Contain-Acc & LLM-Acc & EM & LLM-Acc \\
\midrule
Zero-shot & 0.422 & 0.451 & 0.300 &
0.505 & 0.398 & 0.345 & 0.269 \\
Standard-RAG & 0.700 & 0.702 & 0.460 &
0.689 & 0.617 & 0.467 & 0.669 \\
HippoRAG & 0.690 & 0.835 & 0.609 &
0.555 & 0.575 & 0.476 & 0.823 \\
Cog-RAG & 0.822 & 0.843 & 0.562 &
0.768 & 0.700 & 0.554 & 0.831 \\
Hyper-RAG & 0.735 & 0.808 & 0.511 &
0.785 & 0.745 & 0.541 & 0.808 \\
LightRAG & 0.861 & 0.877 & 0.600 &
0.821 & 0.674 & 0.518 & 0.877 \\
LinearRAG & 0.745 & 0.887 & 0.649 &
0.810 & 0.850 & \textbf{0.688} & 0.762 \\
\textbf{STITCH-RAG} & \textbf{0.900} & \textbf{0.895} &
\textbf{0.654} & \textbf{0.877} & \textbf{0.861} &
0.670 & \textbf{0.884} \\
\bottomrule
\end{tabular}}
\end{table*}
STITCH-RAG ranks first in Contain-Acc and LLM-Acc on every dataset in Table~\ref{tab:q1}. Relative to the strongest per-dataset baseline, the Contain-Acc/LLM-Acc margins are 3.9/0.8 points on HotpotQA and 5.6/1.1 points on 2Wiki. The Mix LLM-Acc margin is 0.7 points. LinearRAG is the closest structural comparison, and STITCH-RAG exceeds it by 15.5 and 6.7 Contain-Acc points on HotpotQA and 2Wiki, respectively. These aggregate differences do not isolate the responsible component. LinearRAG has higher 2Wiki EM (0.688 versus 0.670); without answer-length statistics or error analysis, we report this as a limitation rather than attribute it to response verbosity. Contain-Acc and Recall@8 provide the strongest objective evidence. EM is mixed, and  the smaller LLM-Acc margins remain auxiliary.

\begin{table}[t]
\centering
\caption{Passage-level Recall@8 on HotpotQA and 2Wiki,
defined as $|G_q \cap R_q^8| / |G_q|$ where $G_q$ is the
ground-truth supporting chunk set and $R_q^8$ is the top-8
retrieved set. Mix is excluded because supporting-fact
annotations are unavailable. Only methods with publicly
available retrieval outputs under the standardized embedding
protocol are included.}
\label{tab:retrieval_recall}
\resizebox{0.88\linewidth}{!}{
\begin{tabular}{lcc}
\toprule
Method & HotpotQA & 2Wiki \\
\midrule
Standard-RAG & 0.612 & 0.587 \\
LightRAG & 0.741 & 0.693 \\
LinearRAG & 0.723 & 0.712 \\
\textbf{STITCH-RAG} & \textbf{0.798} & \textbf{0.769} \\
\bottomrule
\end{tabular}}
\end{table}

Among methods with standardized retrieval outputs, STITCH-RAG exceeds LinearRAG in Recall@8 by 7.5 points on HotpotQA and 5.7 points on 2Wiki. This pattern is consistent with broader retrieval coverage, but the comparison omits methods without standardized outputs and does not show that retrieval alone causes the answer-accuracy gains. The margins also match the mechanism in Proposition~\ref{thm:ppr_evidence_mass}, though the experiment does not verify its alignment condition because $\gamma$ is not measured per query.

\subsection{Answer Quality Comparison (RQ2)}

We apply the LightRAG pairwise protocol as an auxiliary preference-based evaluation of comprehensiveness, diversity, and empowerment. STITCH-RAG exceeds the 50\% win-rate line on most dataset--dimension pairs (Figure~\ref{q2}, Appendix~\ref{app:additional_results}), with the largest margins in comprehensiveness and diversity on HotpotQA and 2Wiki, consistent with the Recall@8 results in Table~\ref{tab:retrieval_recall}. On Mix, LightRAG remains competitive in selected dimensions, plausibly because explicit relation extraction can better encode structured domain knowledge when topic-hyperedge boundaries are less distinct. Both rely on an LLM judge and remain exploratory.

\subsection{Ablation Study (RQ3)}
\label{sec:ablation}

\begin{table*}[!t]
\centering
\caption{Ablation of STITCH-RAG (three-run LLM-Acc average).
The upper block removes STIBP/PPR or uses dense retrieval; the lower
block removes one STIBP channel or replaces continuous priors with
binary initialization.}
\label{tab:ablation_extended}
\resizebox{0.96\textwidth}{!}{
\begin{tabular}{lccc}
\toprule
Variant & HotpotQA LLM-Acc & 2Wiki LLM-Acc
& Mix LLM-Acc \\
\midrule
\textbf{Full STITCH-RAG} & \textbf{0.895} & \textbf{0.861}
& \textbf{0.884} \\
w/o STIBP & 0.823 & 0.707 & 0.800 \\
w/o PPR & 0.825 & 0.793 & 0.854 \\
w/o all (dense retrieval) & 0.761 & 0.684 & 0.672 \\
\midrule
w/o spatial bridging &       0.838     &      0.742        &  0.814\\
w/o index-proximity bridging & 0.840 & 0.748 & 0.815\\
Strict binary-only initialization diagnostic
&      0.459      &        0.365      &  0.469\\
\bottomrule
\end{tabular}}
\end{table*}

\begin{table*}[!t]
\centering
\caption{Efficiency and LLM-Acc on Mix. Indexing costs are one-time;
query-stage tokens include answer generation and method-specific
query-time LLM calls under identical hardware and API conditions.}
\label{q4}
\resizebox{0.96\textwidth}{!}{\begin{tabular}{lrrrrrrr}
\toprule
\textbf{Method} & \multicolumn{2}{c}{\textbf{Time (s)}}
& \multicolumn{4}{c}{\textbf{Token Consumption}}
& \textbf{LLM-Acc} \\
\cmidrule(lr){2-3}\cmidrule(lr){4-7}
& \textbf{Indexing} & \textbf{Retrieval}
& \textbf{Idx. Prompt} & \textbf{Idx. Completion}
& {\textbf{Query Prompt}} &
{\textbf{Query Completion}} & \\
\midrule
HippoRAG & 1706.91 & 39.56 & 1382281 & 1321272 &
55845.91 & 3774.76 & 0.823 \\
Cog-RAG & 19109.94 & 126.88 & 3082137 & 7988713 &
36840.47 & 13246.30 & 0.831 \\
Hyper-RAG & 70414.64 & 52.20 & 4652473 & 7264711 &
18557.08 & 5521.35 & 0.808 \\
LightRAG & 89882.40 & 75.87 & 6588007 & 8685853 &
29795.00 & 8148.52 & 0.877 \\
LinearRAG & 633.37 & 23.00 & 0 & 0 &
6395.72 & 261.50 & 0.762 \\
STITCH-RAG & 10384.41 & 31.12 & 1229660 & 14196127 &
4054.52 & 2558.72 & 0.884 \\
\bottomrule
\end{tabular}}
\end{table*}
Table~\ref{tab:ablation_extended} provides controlled evidence that STIBP and PPR contribute non-redundant gains: removing STIBP/PPR reduces HotpotQA LLM-Acc by 7.2/7.0 points and 2Wiki LLM-Acc by 15.4/6.8 points. The larger STIBP loss on 2Wiki is consistent with stronger cross-chunk entity-tracking demands. Under the fixed pipeline, semi-merging achieves higher point estimates than both controlled merge alternatives (Appendix~\ref{app:additional_results}), although this diagnostic does not establish general optimality. The lower block shows that both bridging channels and frequency-adaptive decay are beneficial; adaptive decay reaches 0.884 on Mix, compared with 0.853 and 0.846 for fixed and logarithmic scaling. Each variant retains the remaining pipeline components and changes only the stated operation. Appendix~\ref{app:additional_results} details the binary-initialization diagnostic and its controls, while Appendix~\ref{app:implementation} provides the extraction prompt and phase-level visualization.

\subsubsection{Hyperparameter Sensitivity Analysis}
\label{sec:hyper}
We select $\delta=0.5$, $\lambda=0.3$, and $k=8$ on the HotpotQA development set and transfer them unchanged to 2Wiki and Mix. Appendix~\ref{app:sensitivity} reports the complete sweeps; among the tested decay rules, adaptive $\alpha$ attains the highest reported Mix LLM-Acc.

\subsection{Efficiency Analysis (RQ4)}
Table~\ref{q4} pairs STITCH-RAG's top Mix LLM-Acc with the second-lowest structural-retrieval latency. Its added cost occurs during one-time indexing of context-aware entity descriptions, whereas structural retrieval issues no LLM API calls. Appendix~\ref{app:complexity} gives the component-level cost decomposition.

\section{Conclusion}

STITCH-RAG couples topic-preserving hyperedges, semi-merged local entity states, frequency-adaptive STIBP, and continuous-prior PPR. Proposition~\ref{prop:topic_nonidentifiability} applies only to provenance-discarding pairwise inputs, and Proposition~\ref{prop:semimerge_recoverability} compares three controlled merge constructions. The decay and PPR results characterize the proposed mechanisms; they are not end-to-end optimality guarantees.

Under the reported protocol, STITCH-RAG attains the highest Contain-Acc and LLM-Acc point estimates in Table~\ref{tab:q1} and the highest Recall@8 in Table~\ref{tab:retrieval_recall}. The merge and binary diagnostics remain inconclusive because baseline variance and per-query alignment are unmeasured, 2Wiki EM is below LinearRAG, and the merge comparison does not establish a compression--relevance optimum.

Current limitations include extraction and exact-name-linking errors, preprocessing-order-sensitive index proximity, unmeasured baseline variance, same-family judging on Mix, and the approximately 15M-token Mix indexing cost. Future work will study robust entity linking, order-insensitive propagation, adaptive routing~\cite{jeong2024adaptive}, incremental indexing, and multimodal topic hyperedges.

\bibliography{aaai2027}
\vfill
\cleardoublepage

\section*{Reproducibility Checklist}
\paragraph{This paper:}
\begin{itemize}
	\item Includes a conceptual outline and/or pseudocode description of AI methods introduced (yes/partial/no/NA) {\bf yes}
	\item Clearly delineates statements that are opinions, hypothesis, and speculation from objective facts and results (yes/no) {\bf yes}
	\item Provides well-marked pedagogical references for less-familiar readers to gain background necessary to replicate the paper (yes/no) {\bf yes}
\end{itemize}

\paragraph{Does this paper make theoretical contributions? (yes/no)} {\bf yes}

If yes, please complete the list below.
\begin{itemize}
	\item All assumptions and restrictions are stated clearly and formally. (yes/partial/no) {\bf yes}
	\item All novel claims are stated formally (e.g., in theorem statements). (yes/partial/no) {\bf yes}
	\item Proofs of all novel claims are included. (yes/partial/no) {\bf yes}
	\item Proof sketches or intuitions are given for complex and/or novel results. (yes/partial/no) {\bf yes}
	\item Appropriate citations to theoretical tools used are given. (yes/partial/no) {\bf yes}
	\item All theoretical claims are demonstrated empirically to hold. (yes/partial/no/NA) {\bf partial}
	\item All experimental code used to eliminate or disprove claims is included. (yes/no/NA) {\bf yes}
\end{itemize}

\paragraph{Does this paper rely on one or more datasets? (yes/no)} {\bf yes}

If yes, please complete the list below.
\begin{itemize}
	\item A motivation is given for why the experiments are conducted on the selected datasets (yes/partial/no/NA) {\bf yes}
	\item All novel datasets introduced in this paper are included in a data appendix. (yes/partial/no/NA) {\bf yes}
	\item All novel datasets introduced in this paper will be made publicly available upon publication of the paper with a license that allows free usage for research purposes. (yes/partial/no/NA) {\bf yes}
	\item All datasets drawn from the existing literature (potentially including authors' own previously published work) are accompanied by appropriate citations. (yes/no/NA) {\bf yes}
	\item All datasets drawn from the existing literature (potentially including authors' own previously published work) are publicly available. (yes/partial/no/NA) {\bf yes}
	\item All datasets that are not publicly available are described in detail, with explanation why publicly available alternatives are not scientifically satisficing. (yes/partial/no/NA) {\bf NA}
\end{itemize}

\paragraph{Does this paper include computational experiments? (yes/no)} {\bf yes}

If yes, please complete the list below.
\begin{itemize}
	\item This paper states the number and range of values tried per (hyper-) parameter during development of the paper, along with the criterion used for selecting the final parameter setting. (yes/partial/no/NA) {\bf yes}
	\item Any code required for pre-processing data is included in the appendix. (yes/partial/no) {\bf no}
	\item All source code required for conducting and analyzing the experiments is included in a code appendix. (yes/partial/no) {\bf no}
	\item All source code required for conducting and analyzing the experiments will be made publicly available upon publication of the paper with a license that allows free usage for research purposes. (yes/partial/no) {\bf yes}
	\item All source code implementing new methods have comments detailing the implementation, with references to the paper where each step comes from. (yes/partial/no) {\bf yes}
	\item If an algorithm depends on randomness, then the method used for setting seeds is described in a way sufficient to allow replication of results. (yes/partial/no/NA) {\bf NA}
	\item This paper specifies the computing infrastructure used for running experiments (hardware and software), including GPU/CPU models; amount of memory; operating system; names and versions of relevant software libraries and frameworks. (yes/partial/no) {\bf partial}
	\item This paper formally describes evaluation metrics used and explains the motivation for choosing these metrics. (yes/partial/no) {\bf yes}
	\item This paper states the number of algorithm runs used to compute each reported result. (yes/no) {\bf no}
	\item Analysis of experiments goes beyond single-dimensional summaries of performance (e.g., average; median) to include measures of variation, confidence, or other distributional information. (yes/no) {\bf yes}
	\item The significance of any improvement or decrease in performance is judged using appropriate statistical tests (e.g., Wilcoxon signed-rank). (yes/partial/no) {\bf no}
	\item This paper lists all final (hyper-)parameters used for each model/algorithm in the paper's experiments. (yes/partial/no/NA) {\bf yes}
\end{itemize}

\vfill
\cleardoublepage
\thispagestyle{empty}

\section*{Appendices}
\renewcommand{\thesection}{\Alph{section}}
\setcounter{section}{0}

\section{Reproducibility and Implementation Details}
\label{app:implementation}

\noindent\textbf{Dataset sample
composition.} The stratified 1{,}000-question samples contain approximately 520 bridge and 480 comparison questions in HotpotQA, and approximately 280 comparison, 250 inference, 240 compositional, and 230 bridge questions in 2Wiki. Mix is reported only in aggregate, so its results are exploratory. A complete release should document domain counts, corpus and chunk sizes, question and reference-answer provenance, construction and deduplication procedures, train/test isolation checks, and data licenses.

\noindent\textbf{Reproducibility of baselines.}
For LightRAG\footnote{\url{https://github.com/HKUDS/LightRAG}}, Cog-RAG\footnote{\url{https://github.com/haoohu/Cog-RAG}}, Hyper-RAG\footnote{\url{https://github.com/iMoonLab/Hyper-RAG}}, and LinearRAG\footnote{\url{https://github.com/DEEP-PolyU/LinearRAG}}, we use the authors' released implementations and replace only the embedding model with \texttt{text-embedding-v4} to standardize vector representations.

\noindent\textbf{Baseline verification protocol.}
For Cog-RAG and Hyper-RAG, we first reproduce the reported results under the original evaluation settings, including the default embedding models and generation backbones. All reproduced metrics fall within 1.5 percentage points of the reported values; the remaining deviations may reflect API-version or random-seed differences. We then rerun every baseline under the standardized protocol in Section~\ref{sec:metrics}.

\noindent\textbf{Reviewer access and verification.}
The supplementary material provides frozen code snapshots with written sharing permission, checksummed Docker images with pinned dependencies, retrieved chunks and generated answers for all 1{,}000 questions on each dataset, and a comparison log for the reproduction check. The anonymous repository includes wrapper scripts, evaluation scripts, and \texttt{VERIFICATION.md} with checksums and expected metric outputs.

\noindent\textbf{Complete extraction prompt $p_{ext}$.}
The prompt below extracts topic summaries and chunk-conditioned entity descriptions. It is applied independently to each chunk and receives no cross-chunk context.
\begin{lstlisting}[caption={Extraction prompt $p_{ext}$},label=lst:prompt]
---Goal---
Extract the knowledge expressed in the task text. Divide the text into self-contained knowledge segments. For each segment, return:
- reference: the supporting text span.
- knowledge: one self-contained sentence describing the segment.
- entities: an object mapping each entity's complete name to its context-specific description.
Notes:
- Do not omit information stated in the original text.
- Each knowledge sentence must be understandable without additional context.
- Return only valid JSON that can be parsed by JSON.parse().
- If an entity description is absent or only repeats the name, use an empty string ("").
- Ground every description in the task text.
- Use the complete entity name as the key and include that complete name in its non-empty description.

---Output JSON format---
Return one or more objects in a JSON array following this structure:
[
  {
    "reference": "Original supporting text",
    "knowledge": "A self-contained knowledge statement.",
    "entities": {
      "Entity 1": "Entity 1 is described in the task text.",
      "Entity 2": "Entity 2 is described in the task text."
    }
  }
]

######################
Task text:
{content}
Output:
\end{lstlisting}
Each extracted knowledge segment creates one topic-summary hyperedge $e_{k,r}$, whose text is the segment's knowledge sentence. Each entry in its entities field creates an incident entity-state node $v=(v^{name},v^{des})$.

\noindent\textbf{Extraction convention summary.}
The extraction procedure uses three conventions. First, entity boundaries follow the maximal-noun-phrase rule: \emph{New York City} is one entity rather than three tokens. Second, each topic represents one coherent claim or narrative thread; a chunk covering a person's early career and later achievements should yield two topic segments. Third, entity descriptions are grounded in the current chunk, so repeated entities receive role-specific local descriptions.

\section{Algorithmic Details}
\label{app:secB}

\begin{algorithm}[!ht]
\caption{STITCH-RAG: Retrieval Pipeline}
\label{alg:stitch_rag}
\begin{algorithmic}[1]
\REQUIRE Query $q$; semi-merged hypergraph
$\mathcal{H}=(\mathcal{V},\mathcal{E},\phi,\psi)$; chunk
set $C$; parameters $\delta$, $\lambda$, $d$, $k$;
{ANN candidate budgets $K_V,K_E,K_C$}
\ENSURE Top-$k$ retrieved chunks $C_q^k$

\STATE {$(V_q^0,E_q^0,C_q^0)\leftarrow
\operatorname{ANN}(q;\mathcal V,\mathcal E,C,K_V,K_E,K_C)$}
\FOR{each entity node $v \in {V_q^0}$}
    \STATE $A_q(v) \leftarrow \operatorname{sim}(q, v^{des})
    \cdot \mathbf{1}[\operatorname{sim}(q, v^{des}) > \delta]$
\ENDFOR
\FOR{each hyperedge $e \in {E_q^0\cup
\{e:e\ni v,\ v\in V_q^0\}}$}
    \STATE $S_q(e) \leftarrow \operatorname{sim}(q, e)$
\ENDFOR

\STATE Initialize $a^{space}(v) \leftarrow 0$,
$a^{time}(v) \leftarrow 0$ for all $v$
\FOR{each activated node $u$ with $A_q(u) > 0$}
    \FOR{each hyperedge $e \ni u$}
        \FOR{each $v \in V_e \setminus \{u\}$}
            \STATE $a^{space}(v) \leftarrow
            \max\bigl(a^{space}(v),\;
            A_q(u) \cdot S_q(e)\bigr)$
        \ENDFOR
    \ENDFOR
    \STATE $V_u^* \leftarrow \{v \in \mathcal{V} \mid
    \phi(v) = \phi(u)\}$ \COMMENT{name-equivalence group}
    \STATE $\alpha \leftarrow |V_u^*| / \bar{n}$
    \COMMENT{frequency-adaptive decay rate}
    \FOR{each $v \in V_u^* \setminus \{u\}$}
        \STATE $\Delta t \leftarrow |t_u - t_v|$
        \STATE $a^{time}(v) \leftarrow
        \max\bigl(a^{time}(v),\;
        A_q(u) \cdot S_q(e_v) \cdot
        \exp(-\tanh(\alpha \Delta t)) \cdot
        \operatorname{sim}(q, c_v)\bigr)$
    \ENDFOR
\ENDFOR
\STATE $a(v) \leftarrow \operatorname{mean}\bigl(A_q(v),\,
a^{space}(v),\, a^{time}(v)\bigr)$ for all $v$
\STATE {$C_q\leftarrow C_q^0\cup
\{c_v:a(v)>0\}$}

\FOR{each chunk $c \in {C_q}$}
    \STATE ${\rm Score}(c) \leftarrow
    \sum_{v \in \mathcal{V}(c)} a(v) \cdot \ln(1 + N_v)$
    \STATE {$\widetilde I(c) \leftarrow
    \lambda\max\{0,\operatorname{sim}(q,c)\}
    +\ln(1+{\rm Score}(c))$}
\ENDFOR
\STATE {$p_{\mathrm{ST}}(c)\leftarrow
\widetilde I(c)/\sum_{z\in C_q}\widetilde I(z)$ if the sum is
positive; otherwise $p_{\mathrm{ST}}(c)\leftarrow1/|C_q|$}

\STATE Build chunk graph: $\mathcal{B}(c) \leftarrow
\{c' \mid \exists\, u \in \mathcal{V}(c),\,
v \in \mathcal{V}(c') \text{ s.t. } \phi(u) = \phi(v)\}$
\STATE {Column-normalize the adjacency as $P$;
set $P_{cc}=1$ when $\mathcal B(c)=\emptyset$}
\STATE {$\pi^{(0)}\leftarrow p_{\mathrm{ST}}$}
\REPEAT
    \FOR{each chunk $c\in C_q$}
        \STATE {$\pi^{(t)}(c) \leftarrow
        (1-d)p_{\mathrm{ST}}(c)+d\sum_{c'\in C_q}
        P_{cc'}\pi^{(t-1)}(c')$}
    \ENDFOR
\UNTIL{{$\max_c |\pi^{(t)}(c) -
\pi^{(t-1)}(c)| < 10^{-6}$}}

\STATE $C_q^k \leftarrow \operatorname{top\text{-}k}$
chunks by {$\pi^{(T)}(c)$}
\RETURN $C_q^k$
\end{algorithmic}
\end{algorithm}

\begin{algorithm}[!ht]
\caption{STITCH-RAG: Offline Hypergraph Construction}
\label{alg:stitch_index}
\begin{algorithmic}[1]
\REQUIRE Document collection $D$; max chunk size
$c_{\max}$; extraction prompt $p_{ext}$
\ENSURE Semi-merged hypergraph
$\mathcal{H}=(\mathcal{V}, \mathcal{E}, \phi, \psi)$

\STATE Split $D$ into chunks $C = \{c_k\}_{k=1}^{N_C}$
with $|c_k| \leq c_{\max}$
\FOR{each chunk $c_k \in C$}
    \STATE $\{g_{k,r}\}_{r} \leftarrow
    \operatorname{LLM}(c_k \mid p_{ext})$
    \COMMENT{extract (topic, entities) tuples}
    \FOR{each tuple $g_{k,r} = (e_{k,r}, V_{e_{k,r}})$}
        \STATE Add hyperedge $e_{k,r}$ to $\mathcal{E}$
        with topic summary text
        \FOR{each entity $(v^{name}, v^{des})
        \in V_{e_{k,r}}$}
            \STATE Create node $v$ with
            $\phi(v) \leftarrow v^{name}$,
            $\psi(v) \leftarrow v^{des}$
            \STATE Add $v$ to $\mathcal{V}$; record
            incidence $(v, e_{k,r})$ and chunk
            membership $(v, c_k)$
        \ENDFOR
    \ENDFOR
\ENDFOR

\STATE Group nodes by canonical name:
$\{G_s = \{v \mid \phi(v) = s\}\}_{s \in \Sigma}$
\STATE {Sort nodes within each group by the
deterministic preprocessing index $t_v$}

\STATE Embed all entity descriptions $\{v^{des}\}$,
topic summaries $\{e\}$, and chunks $\{c_k\}$

\RETURN $\mathcal{H} = (\mathcal{V}, \mathcal{E},
\phi, \psi)$
\end{algorithmic}
\end{algorithm}

\subsection{Complexity Analysis}
\label{app:complexity}

Let $N_C$, $N_V=|\mathcal{V}|$, $N_E=|\mathcal{E}|$, and $M=\sum_{e\in\mathcal{E}}|e|$ denote the numbers of chunks, entity-state nodes, topic hyperedges, and entity--hyperedge incidence links. The symbolic complexities exclude LLM inference and embedding execution, which are model-dependent and not assumed to be shared across systems. Table~\ref{q4} reports measured time and token consumption separately.

\noindent\textbf{Offline indexing.} Topic incidence storage
requires $O(M)$ links; $\phi$-based grouping requires
$O(N_V)$ hash operations; within-group sorting by chunk index
requires $O(\sum_g n_g\log n_g)$ where $n_g$ is the size of
group $g$. The total symbolic indexing cost is
\[
O\!\left(N_C+N_V+N_E+M+\textstyle\sum_g n_g\log n_g\right),
\]
and storage is $O(N_C+N_V+N_E+M)$. All indexing costs are
incurred once.

\noindent\textbf{Online retrieval.} Let $A_q$, $\mathcal{E}_q$, and $B_q$ denote the activated entity-state set, incident topic hyperedges, and edge set of the query-induced chunk graph. Let $T_{\mathrm{cand}}(N_V,N_E,N_C)$ denote candidate-generation cost over node, hyperedge, and chunk vector indexes. Under exact search, it is $O((N_V+N_E+N_C)d_{\mathrm{emb}})$; ANN cost depends on the selected index and approximation parameters. Spatial propagation costs
$O(\sum_{e\in\mathcal{E}_q}|e|)$, temporal propagation costs
$O(\sum_{u\in A_q}|V_u^*|)$ without materializing pairwise
links, and localized PPR costs $O(T|B_q|)$ for $T$ iterations.
The total online cost is
\[
\begin{aligned}
O\!\biggl(&T_{\mathrm{cand}}(N_V,N_E,N_C)
+\sum_{e\in\mathcal{E}_q}|e|\\
&+\sum_{u\in A_q}|V_u^*|+T|B_q|\biggr).
\end{aligned}
\]
Candidate generation remains corpus-dependent, whereas propagation depends on the activated subgraph. Table~\ref{q4} reports measured latency rather than extrapolating end-to-end scaling from propagation terms alone.

\subsection{Qualitative Compression--Context Interpretation of
Semi-Merging}
\label{app:ib_interpretation}

The three controlled merge strategies trade compact identity representation against local-context retention. Full merging collapses name-equivalent mentions and removes per-chunk descriptions. No merging retains each local description but exposes no cross-chunk name-equivalence relation. Semi-merging retains distinct local states while exposing their operational name equivalence to retrieval.

This interpretation motivates the representation design but is not an information-bottleneck result: the paper neither defines nor estimates mutual information for the three constructions. Table~\ref{tab:merge_extended} is an answer-level controlled comparison and does not directly measure compression or retained task information.

\subsection{Complete Hyperparameter Sensitivity Results}
\label{app:sensitivity}

Tables~\ref{tab:sens_delta} and \ref{tab:sens_lambda} report post-hoc transfer sensitivity of the activation threshold $\delta$ and balance parameter $\lambda$ on Mix. Bold values were chosen on the HotpotQA development set, not on Mix. Table~\ref{tab:sens_k} reports the top-$k$ sweep on that 200-question development set. Each sweep fixes all other hyperparameters and the retrieval pipeline at their defaults ($k=8$, $d=0.85$, adaptive $\alpha=n_u/\bar{n}$).

\begin{table}[h]
\centering
\caption{Post-hoc Mix LLM-Acc sensitivity to
$\delta$. The bold row denotes the value preselected on HotpotQA,
not a value selected on Mix. All other hyperparameters are fixed.}
\label{tab:sens_delta}
\begin{tabular}{lc}
\toprule
$\delta$   & Mix \\
\midrule
0.3 &        0.861         \\
0.4 &          0.815       \\
\textbf{0.5}  &     0.884    \\
0.6 &          0.838      \\
0.7 &        0.861       \\
\bottomrule
\end{tabular}
\end{table}

\begin{table}[h]
\centering
\caption{Post-hoc Mix LLM-Acc sensitivity to
$\lambda$. The bold row denotes the value preselected on HotpotQA,
not a value selected on Mix. All other hyperparameters are fixed.}
\label{tab:sens_lambda}
\begin{tabular}{lc}
\toprule
$\lambda$ & Mix \\
\midrule
0.1 &   0.823    \\
\textbf{0.3} &   0.884  \\
0.5 &     0.846  \\
0.7 &   0.846    \\
0.9 &    0.861   \\
\bottomrule
\end{tabular}
\end{table}

\begin{table}[h]
\centering
\caption{LLM-Acc under varying top-$k$
retrieval count on the 200-question HotpotQA development
set. The bold column denotes the selected value. All other
hyperparameters are fixed at their defaults.}
\label{tab:sens_k}
\resizebox{\linewidth}{!}{%
\begin{tabular}{lcccccc}
\toprule
$k$ & 5 & 6 & 7 & \textbf{8} & 9 & 10 \\
\midrule
LLM-Acc & 0.823 & 0.853 & 0.861 & \textbf{0.884}
& 0.869 & 0.876 \\
\bottomrule
\end{tabular}
}
\end{table}

\section{Additional Experimental Results}
\label{app:additional_results}

This appendix contains the phase-level ablation visualization, pairwise answer-quality comparison, merge-strategy diagnostic, and temporal-decay comparison cited in the main text.

\begin{figure*}[t]
\centering
\includegraphics[width=\textwidth]{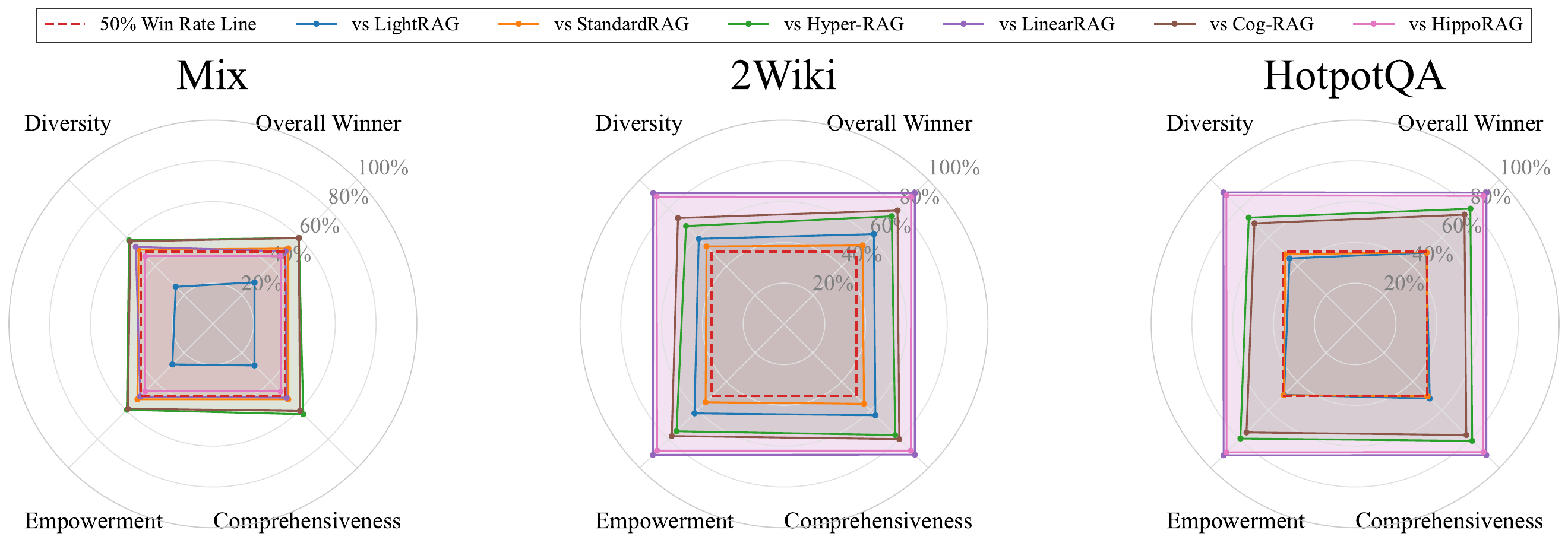}
\caption{Pairwise answer-quality win rates
of STITCH-RAG against five baselines on HotpotQA, 2Wiki,
and Mix, judged by Qwen-Max on comprehensiveness (coverage
of relevant details), diversity (variety of useful
perspectives), and empowerment (utility for informed
judgment). The dashed circle marks the 50\% win-rate line.
Values outside it favor STITCH-RAG.}
\label{q2}
\end{figure*}

\begin{table}[t]
\centering
\caption{LLM-Acc point estimates for one
implementation of each merge strategy under a fixed retrieval
pipeline. The diagnostic does not evaluate run-level variance,
retrieval recall, graph compression, statistical significance, or
all possible implementations of the three strategies.}
\label{tab:merge_extended}
\begin{tabular}{lccc}
\toprule
Strategy & HotpotQA & 2Wiki & Mix \\
\midrule
Full-Merge & 0.873 & 0.839 & 0.807 \\
Semi-Merge & \textbf{0.895} & \textbf{0.861}
& \textbf{0.884} \\
No-Merge  & 0.879 & 0.844 & 0.831 \\
\bottomrule
\end{tabular}
\end{table}

\begin{table}[t]
\centering
\caption{LLM-Acc on Mix under three temporal
decay functions, with all other pipeline components held
fixed (same hypergraph, same PPR initialization, same
$\alpha = n_u/\bar{n}$).}
\label{tab:decay_comparison}
\begin{tabular}{lc}
\toprule
Decay Function & LLM-Acc \\
\midrule
$\exp(-\alpha\Delta t)$ (pure exponential) & 0.853 \\
$\sigma(-\alpha\Delta t+b)$ (sigmoid gate) & 0.815 \\
$\exp(-\tanh(\alpha\Delta t))$ (ours) & \textbf{0.884} \\
\bottomrule
\end{tabular}
\end{table}

\section{Proofs of Main Results}
\label{app:SecD}

This appendix states the definitions used by Propositions~\ref{prop:topic_nonidentifiability} and~\ref{prop:semimerge_recoverability}, then provides all proofs.

\begin{definition}[Pairwise Topic Projection]
\label{def:pairwise_projection}
Given a topic hypergraph $\mathcal{H}=(\mathcal{V},\mathcal{E})$,
its pairwise projection is the graph
\[
\begin{aligned}
\Pi(\mathcal{H})&=(\mathcal{V},\mathcal{E}_2),\\
\mathcal{E}_2&=\bigl\{\{u,v\}\mid u\neq v,\;
\exists e\in\mathcal{E}\ \text{s.t.}\ u,v\in e\bigr\}.
\end{aligned}
\]
The projection retains pairwise co-occurrence but discards
the identity of the topic hyperedge that generated each pair.
\end{definition}

\begin{definition}[Semi-Merged Hypergraph]
\label{def:semi_merge}
Let $\mathcal{H}=(\mathcal{V},\mathcal{E},\phi,\psi)$ denote
a semi-merged hypergraph, where $\mathcal{V}$ is the set of
entity nodes, $\mathcal{E}$ is the set of hyperedges,
$\phi:\mathcal{V}\to\Sigma$ maps each node to an entity name
in name space $\Sigma$, and $\psi:\mathcal{V}\to\mathcal{D}$
maps each node to a context-specific description in description
space $\mathcal{D}$. Two nodes $v_i,v_j\in\mathcal{V}$ satisfy
$\phi(v_i)=\phi(v_j)$ when their normalized canonical-name strings
are identical. This operational relation approximates, but does not
guarantee, real-world co-reference: aliases can produce false splits,
whereas homonyms can produce false links. The equivalence classes
under $\phi$ are called \textit{name-equivalence groups}.
\end{definition}

\begin{definition}[Full-Merge and No-Merge Hypergraphs]
\label{def:full_no_merge}
A \textit{full-merge} hypergraph $\mathcal{H}_F$ collapses
all nodes within each name-equivalence group into a single node
and discards per-chunk descriptions. A \textit{no-merge}
hypergraph $\mathcal{H}_N$ treats every extracted entity
mention as an independent node with no cross-chunk identity
linkage; equivalently, $\mathcal{H}_N$ does not define a map
$\phi$ across chunks.
\end{definition}

\begin{proof}[\bf Proof of Proposition~\ref{prop:topic_nonidentifiability}]
We construct two topic hypergraphs with identical pairwise
projections and exhibit a query under which topic-selective
spatial propagation produces different influence assignments,
establishing the non-identifiability claim.

Let the entity set be $\mathcal{V}=\{a,b,c,d\}$. Consider two hypergraphs:
\[
\mathcal{H}_1=\bigl(\mathcal{V},\{\{a,b,c\},\{a,b,d\},\{a,c,d\},\{b,c,d\}\}\bigr)
\]
and
\[
\mathcal{H}_2=\bigl(\mathcal{V},\{\{a,b,c,d\}\}\bigr).
\]
Every pair from $\{a,b,c,d\}$ co-occurs in at least one hyperedge
of $\mathcal{H}_1$ (each triple contains all $\binom{3}{2}=3$ pairs,
and the four triples together cover all $\binom{4}{2}=6$ pairs),
and every pair co-occurs in the single hyperedge of $\mathcal{H}_2$.
Hence both hypergraphs project to the complete graph on four
vertices:
\[
\Pi(\mathcal{H}_1)=\Pi(\mathcal{H}_2)=K_4.
\]

Fix a query $q$ that is semantically relevant only to
the topic $\{a,b,c\}$, and let $A_q(a)=1$ be the only activated
source node. In $\mathcal{H}_1$, assign topic relevance scores
\[
S_q(\{a,b,c\})=1,
\qquad
S_q(e)=0
\quad
\text{for all other } e\in\mathcal{E}_1.
\]
The spatial bridging formula $a^{space}(v)=A_q(a)\cdot S_q(e)$
for each $v\in V_e\setminus\{a\}$ gives
\[
a^{space}(b)=1,\qquad a^{space}(c)=1,\qquad a^{space}(d)=0,
\]
since $d$ does not appear in any hyperedge $e$ with $S_q(e)>0$
under $\mathcal{H}_1$.

In $\mathcal{H}_2$, the only available hyperedge is $\{a,b,c,d\}$.
Two exhaustive cases arise. If $S_q(\{a,b,c,d\})>0$, spatial
bridging from $a$ assigns the same positive score $A_q(a)\cdot
S_q(\{a,b,c,d\})$ to each of $b$, $c$, and $d$, so $d$
receives positive influence. If $S_q(\{a,b,c,d\})=0$, no
influence propagates at all, so $b$ and $c$ receive zero
influence. In neither case can $\mathcal{H}_2$ reproduce the
selective outcome $a^{space}(b)>0$, $a^{space}(c)>0$,
$a^{space}(d)=0$.

Since $\Pi(\mathcal{H}_1)=\Pi(\mathcal{H}_2)$, any retrieval
rule whose input is restricted to $\Pi(\mathcal{H})$ must
produce identical outputs on $\mathcal{H}_1$ and $\mathcal{H}_2$,
yet the propagation outcomes differ for the query $q$
constructed above. Therefore, no retrieval rule defined solely
on the pairwise projection can, in general, reproduce
topic-selective propagation over hyperedges.
\end{proof}

\begin{proof}[\bf Proof of Proposition~\ref{prop:semimerge_recoverability}]
We verify the influence received by $v_r$ under each of
the three constructions, holding $q$, $G$, and all activation
scores fixed throughout. The assumptions in force are:
$A_q(v_s)>\delta$; $A_q(v_r)\leq\delta$, so after
thresholding $A_q(v_r)=0$; $v_r$ is not reachable from $v_s$
through any activated topic hyperedge; $S_q(e_{v_r})>0$ for
some hyperedge $e_{v_r}$ containing $v_r$; and
$\operatorname{sim}(q,c_{v_r})>0$.

\medskip
\noindent\textbf{No-merge hypergraph $\mathcal{H}_N$.}
In $\mathcal{H}_N$, the map $\phi$ is not defined across chunks,
so every entity mention is an independent node. The corresponding
name-equivalence group of $v_s$ under the no-merge construction
reduces to the
singleton $V_{v_s}^*=\{v_s\}$, so the temporal bridging loop
does not visit $v_r$. By hypothesis, $v_r$ is not reachable
from $v_s$ through any activated topic hyperedge, so
$a^{space}_{\mathcal{H}_N}(v_r)=0$. Combined with
$A_q(v_r)=0$, the mean aggregation gives
$a_{\mathcal{H}_N}(v_r)=0$. The evidence chunk $c_{v_r}$
therefore receives zero projected influence from $v_r$ and is
structurally unreachable from $v_s$ under $\mathcal{H}_N$.

\medskip
\noindent\textbf{Full-merge hypergraph $\mathcal{H}_F$.}
All nodes in $G$ are collapsed into a single representative
$\bar{v}=\operatorname{merge}(v_1,\ldots,v_m)$ with aggregated
description $\bar{v}^{des}$ and membership in the union of all
original hyperedges and chunks. Two cases arise.

\begin{enumerate}
\item[(2a)] If $A_q(\bar{v})=\operatorname{sim}(q,\bar{v}^{des})>\delta$,
then $\bar{v}$ is activated. Every chunk containing any original
member of $G$ receives the same influence through $\bar{v}$,
because the per-chunk descriptions $\psi(v_i)$ have been
discarded and replaced by the single aggregated description.
The projected score of $c_{v_r}$ under
$\mathcal{H}_F$ therefore equals the projected score of any
other chunk $c_{v_j}$ containing a member of $G$, regardless
of whether $v_j$ is evidence-bearing. The retrieval algorithm
cannot preferentially assign state-specific influence to
$c_{v_r}$ over $c_{v_j}$ using entity-state information.

\item[(2b)] If $A_q(\bar{v})\leq\delta$, then $\bar{v}$ is not
activated and no chunk in $G$'s scope receives entity-level
influence. This arises when the aggregated
description $\bar{v}^{des}$ averages over both query-relevant
and query-irrelevant mentions of the entity, diluting the
cosine similarity below the activation threshold $\delta$
even though $v_s$ would have been activated had its
per-chunk description been retained.
\end{enumerate}

In either case, $\mathcal{H}_F$ cannot assign state-specific
positive influence to $v_r$: in (2a) influence is distributed
uniformly across all $G$-members, and in (2b) it is zero.

\medskip
\noindent\textbf{Semi-merged hypergraph $\mathcal{H}$.}
In $\mathcal{H}$, nodes $v_s$ and $v_r$ retain distinct
descriptions $\psi(v_s)$ and $\psi(v_r)$, while $\phi$
exposes their name-equivalence relationship
$\phi(v_s)=\phi(v_r)$. The name-equivalence group satisfies
$V_{v_s}^*=\{v\in\mathcal{V}\mid\phi(v)=\phi(v_s)\}\ni v_r$.
Since $A_q(v_s)>\delta>0$, node $v_s$ is activated and
temporal bridging iterates over $V_{v_s}^*$. For target
$v_r$ at chunk-index distance $\Delta t=|t_{v_s}-t_{v_r}|$
and frequency-adaptive rate $\alpha=|V_{v_s}^*|/\bar{n}$,
the STIBP formula gives
\[
\begin{aligned}
a^{time}_{v_s}(v_r)
&= A_q(v_s)\cdot S_q(e_{v_r})
\cdot\exp\bigl(-\tanh(\alpha\Delta t)\bigr)\\
&\quad\cdot\operatorname{sim}(q,c_{v_r}).
\end{aligned}
\]
Here $e_{v_r}$ denotes a hyperedge containing $v_r$ with
$S_q(e_{v_r})>0$; such a hyperedge exists by assumption. If
$v_r$ belongs to multiple hyperedges, the max operation in
STIBP ensures the stored $a^{time}(v_r)$ is at least the
source-specific value computed
with this $e_{v_r}$, so it suffices to verify positivity for
one such hyperedge. Each factor is strictly positive:
$A_q(v_s)>\delta>0$ by hypothesis; $S_q(e_{v_r})>0$ by
assumption; $\exp(-\tanh(\alpha\Delta t))>0$ because
$\tanh(x)<1$ for every finite $x$ and the exponential is
always positive; and $\operatorname{sim}(q,c_{v_r})>0$ by
assumption. Hence $a^{time}_{v_s}(v_r)>0$ and therefore
$a^{time}(v_r)>0$.

Since $A_q(v_r)=0$ after thresholding and
$a^{space}(v_r)=0$ by hypothesis, the mean aggregation gives
\begin{align*}
a(v_r)
&=\operatorname{mean}\!\bigl(
  \underbrace{A_q(v_r)}_{=\,0},\;
  \underbrace{a^{space}(v_r)}_{=\,0},\;
  \underbrace{a^{time}(v_r)}_{>\,0}
\bigr)\\
&=\tfrac{1}{3}\,a^{time}(v_r)>0.
\end{align*}
Because $a(v_r)>0$, chunk $c_{v_r}$ receives strictly
positive projected influence via $v_r$ and enters the PPR
initialization with non-zero prior weight.

\medskip
\noindent\textbf{Recoverability separation.}
Combining the three cases,
\[
\begin{aligned}
a_{\mathcal{H}_N}(v_r)&=0,\\
a_{\mathcal{H}_F}(v_r)&\text{ is undifferentiated or zero},\\
a_{\mathcal{H}}(v_r)&>0.
\end{aligned}
\]
Among these three constructions, the semi-merged hypergraph is the only one
satisfying both requirements simultaneously: $\phi$ exposes
the name-equivalence relation needed for temporal bridging, and
$\psi$ retains the per-chunk descriptions needed for
state-specific activation, completing the proof.
\end{proof}

\begin{proof}[\bf Proof of Theorem~\ref{thm:stibp_snr}]
We establish a lower bound for the contributions propagated
from one fixed activated source $u$ by bounding
$\max_{r\in\mathcal{R}}a^{time}_u(r)$ from below and
$\sum_{n\in\mathcal{N}}a^{time}_u(n)$ from above. Throughout,
$A_q(u)=a>0$, $\alpha=|V_u^*|/\bar{n}>0$, and the nonempty
sets $\mathcal{R}$ and $\mathcal{N}$ satisfy the conditions in
the theorem. Because $g_n>0$ for every $n\in\mathcal N$, the
denominator is positive.

\medskip
\noindent\textbf{Lower bound for relevant targets.}
For any $r\in\mathcal{R}$, the definition of $a^{time}_u(r)$
and the assumption $g_r=S_q(e_r)\operatorname{sim}(q,c_r)\geq\beta$ give
\begin{align*}
a^{time}_u(r)
&= a\,g_r\exp(-\tanh(\alpha\Delta_r))\\
&\geq a\beta\exp(-\tanh(\alpha\Delta_r)).
\end{align*}
Since $\tanh(x)\in[0,1)$ for all finite $x\geq0$, we have
$-\tanh(\alpha\Delta_r)>-1$, so
$\exp(-\tanh(\alpha\Delta_r))>e^{-1}$ for every finite
$\Delta_r$. Here $\Delta_r$ is a finite non-negative
integer because chunk indices are bounded, and $\alpha>0$ by
assumption. Therefore
\[
a^{time}_u(r)\geq a\beta e^{-1}
\qquad\text{for all }r\in\mathcal{R},
\]
and in particular
\[
\max_{r\in\mathcal{R}}a^{time}_u(r)\geq a\beta e^{-1}.
\]

\medskip
\noindent\textbf{Upper bound for noisy targets.}
For any $n\in\mathcal{N}$, the assumptions $0<g_n\leq\epsilon$
and $\Delta_n\geq\tau>0$ give
\begin{align*}
a^{time}_u(n)
&= a\,g_n\exp(-\tanh(\alpha\Delta_n))\\
&\leq a\epsilon\exp(-\tanh(\alpha\Delta_n)).
\end{align*}
The function $x\mapsto\exp(-\tanh(\alpha x))$ is
non-increasing for $x\geq0$ when $\alpha>0$, because
$\tanh(\alpha x)$ is non-decreasing in $x$. Since
$\Delta_n\geq\tau$, we obtain
$\exp(-\tanh(\alpha\Delta_n))\leq\exp(-\tanh(\alpha\tau))$,
and hence
\[
a^{time}_u(n)\leq a\epsilon\exp(-\tanh(\alpha\tau))
\qquad\text{for all }n\in\mathcal{N}.
\]
Summing over $\mathcal{N}$:
\[
\sum_{n\in\mathcal{N}}a^{time}_u(n)
\leq a|\mathcal{N}|\epsilon\exp(-\tanh(\alpha\tau)).
\]

\medskip
\noindent\textbf{Signal-to-noise ratio.}
Combining the two bounds:
\begin{align*}
\frac{\max_{r\in\mathcal{R}}a^{time}_u(r)}
{\sum_{n\in\mathcal{N}}a^{time}_u(n)}
&\geq
\frac{a\beta e^{-1}}
{a|\mathcal{N}|\epsilon\exp(-\tanh(\alpha\tau))}\\
&=
\frac{\beta}{|\mathcal{N}|\epsilon}
\exp(\tanh(\alpha\tau)-1).
\end{align*}

\medskip
\noindent\textbf{Monotonicity in $\alpha$.}
For fixed $\mathcal R$, $\mathcal N$, $\beta$, $\epsilon$,
$\tau$, and $\bar n$, the function
$\alpha\mapsto\tanh(\alpha\tau)$ is strictly increasing for
$\alpha\geq0$. Therefore
$\exp(\tanh(\alpha\tau)-1)$ is non-decreasing in $\alpha$.
This statement holds with the sets and constants fixed and does
not compare realized ratios across name-equivalence groups of
different sizes.
\end{proof}

\begin{proof}[\bf Proof of Proposition~\ref{thm:ppr_evidence_mass}]
The proof proceeds by showing that PPR evidence mass is
linear in the initialization prior, then applying linearity
to compare $p_{\mathrm{ST}}$ and $p_{\mathrm{B}}$ directly.

\medskip
\noindent\textbf{Linearity of PPR evidence mass.}
The PPR stationary distribution $\pi(p)$ under prior $p$ is
the unique fixed point of $\pi=(1-d)p+dP\pi$. Solving by
Neumann series, which converges absolutely because $d<1$ and
$P$ is column-stochastic with $\|P\|_1=1$:
\[
\pi(p)=(1-d)\sum_{\ell=0}^{\infty}d^\ell P^\ell p.
\]
The series converges in $\ell^1$ norm because
$(1-d)\sum_{\ell=0}^{\infty}d^\ell\|P^\ell p\|_1\leq
(1-d)\sum_{\ell=0}^{\infty}d^\ell\|p\|_1=\|p\|_1<\infty$
for any normalized prior $p$, justifying the interchange of
summation and inner product below. Therefore,
\begin{align*}
M_{\mathcal{R}}(p)
&= \mathbf{1}_{\mathcal{R}}^\top\pi(p)
 = \mathbf{1}_{\mathcal{R}}^\top(1-d)\sum_{\ell=0}^{\infty}d^\ell P^\ell p\\
 &= (1-d)\sum_{\ell=0}^{\infty}d^\ell\,\mathbf{1}_{\mathcal{R}}^\top P^\ell p.
\end{align*}
For each $\ell$, the scalar $\mathbf{1}_{\mathcal{R}}^\top P^\ell p$
equals $((P^\ell)^\top\mathbf{1}_{\mathcal{R}})^\top p$ because
$u^\top Av = (A^\top u)^\top v$ for any compatible vectors
and matrix. Hence,
\[
M_{\mathcal{R}}(p)
= \left((1-d)\sum_{\ell=0}^{\infty}d^\ell(P^\ell)^\top\mathbf{1}_{\mathcal{R}}\right)^\top p
= h_{\mathcal{R}}^\top p,
\]
where the interchange of the finite linear functional
$(\cdot)^\top p$ and the absolutely convergent series is
valid by the dominated convergence theorem for sums.

\medskip
\noindent\textbf{Evidence-mass comparison.}
Since $M_{\mathcal{R}}(p)=h_{\mathcal{R}}^\top p$ is linear in $p$,
\[
M_{\mathcal{R}}(p_{\mathrm{ST}})-M_{\mathcal{R}}(p_{\mathrm{B}})
= h_{\mathcal{R}}^\top(p_{\mathrm{ST}}-p_{\mathrm{B}}).
\]
The condition $h_{\mathcal{R}}^\top(p_{\mathrm{ST}}-p_{\mathrm{B}})\geq\gamma>0$
yields immediately
\[
M_{\mathcal{R}}(p_{\mathrm{ST}})\geq M_{\mathcal{R}}(p_{\mathrm{B}})+\gamma.
\]

\medskip
\noindent\textbf{Interpretation of $h_{\mathcal{R}}$.}
The $c$-th entry of $h_{\mathcal{R}}$ equals
$(1-d)\sum_{\ell=0}^\infty d^\ell[(P^\ell)^\top\mathbf{1}_\mathcal{R}]_c$,
the probability that a forward walk initialized at chunk $c$
occupies $\mathcal R$ when stopped after a geometrically
distributed number of transitions. The condition
$h_{\mathcal{R}}^\top(p_{\mathrm{ST}}-p_{\mathrm{B}})>0$
therefore holds whenever STIBP places more prior mass
on chunks with high discounted reachability to
$\mathcal{R}$ than binary initialization does. The
aggregate ablation in Table~\ref{tab:ablation_extended} is
consistent with this mechanism, but it does not establish the
alignment condition for every query.
\end{proof}

\medskip
\noindent\textbf{Prior and dangling-node convention.}
The retrieval procedure directly uses the
non-negative normalized prior
$p_{\mathrm{ST}}(c)=\widetilde I(c)/\sum_z\widetilde I(z)$
when the denominator is positive, with the same uniform fallback
used in the main text otherwise.
For every dangling chunk $c$, the self-loop $P_{cc}=1$ makes
$P$ column stochastic. Thus the assumptions used in the proposition,
the main-text iteration, and Algorithm~\ref{alg:stitch_rag} are
identical; no post-hoc normalization argument is required.

\medskip
\noindent\textbf{Geometrically stopped forward-walk
reading of $h_{\mathcal{R}}$.}
Let a walk start at chunk $c$ and stop at its current state with
probability $1-d$ before each transition. It stops after $\ell$
transitions with probability $(1-d)d^\ell$. Consequently,
$h_{\mathcal{R}}(c)$ is exactly the probability that the walk
occupies $\mathcal R$ when it stops. This interpretation follows
directly from the discounted occupancy series and does not require
a reverse-time transition process.

\section{Extended Related Work}
\label{app:extended_related}

This appendix expands the two comparisons that determine STITCH-RAG's design: what the index retains about a multi-entity topic and how the retriever initializes propagation. We compare published representations and algorithms; provenance-enriched variants outside those specifications are not ruled out.

\subsection{Representing Topic Provenance and Local Entity State}

GraphRAG extracts entity--relation triples and summarizes graph communities~\cite{edge2024local}. LightRAG reduces this indexing burden with a dual-layer graph~\cite{guo2024lightrag}, and NodeRAG uses heterogeneous node types to retain multiple evidence granularities~\cite{xu2025noderag}. These systems can encode rich pairwise relations. The distinction relevant here is narrower: if retrieval receives an unlabeled projection with no shared event, topic identifier, or provenance, it cannot reconstruct which $n$-ary group generated a pair. Proposition~\ref{prop:topic_nonidentifiability} formalizes this input-specific loss.

LinearRAG avoids LLM-based relation extraction and links canonical entity names for semantic bridging and PPR~\cite{zhuang2025linearrag}. Its published index represents matching names as one graph entity, so it does not expose the chunk-specific states used by STIBP. In STITCH-RAG, $\psi$ retains the local descriptions and $\phi$ links their normalized names. Proposition~\ref{prop:semimerge_recoverability} tests this difference with controlled full-merge, no-merge, and semi-merge constructions.

HyperGraphRAG encodes $n$-ary relational facts, Cog-RAG uses thematic hyperedges, and CogniRAG represents causal structures~\cite{luo2025hypergraphrag,hu2026cog,he2026cognirag}. STITCH-RAG adopts thematic hyperedges but does not collapse all occurrences of an entity name. OKH-RAG also models cross-occurrence structure through learned precedence and sequence inference~\cite{wu2026knowledge}; STITCH-RAG instead uses a training-free index-proximity rule gated by query relevance and frequency under $\phi$.

Chunk granularity changes the number and scope of these local states. Proposition-level retrieval favors narrow evidence units~\cite{chen2024dense}, while RAPTOR retrieves recursively summarized clusters~\cite{sarthi2024raptor}. AutoRAG reports that the preferred granularity depends on question complexity~\cite{kim2025autorag}. STITCH-RAG fixes the chunking policy and bridges the resulting boundaries; it does not determine the optimal chunk size.

\subsection{Initializing and Controlling Graph Propagation}

SubGraphRAG retrieves a query-matched local subgraph~\cite{li2025simple}, and GFM-RAG learns graph retrieval from training data~\cite{luo2025gfmrag}. HippoRAG and HippoRAG2 use PPR for broader diffusion~\cite{gutierrez2024hipporag,gutierrez2025rag}. Their entity-match prior is binary: every matched entity is seeded and every unmatched entity receives zero. STIBP instead derives a continuous prior from the query similarity of local entity states, their topic hyperedges, and their name-equivalent neighbors. Proposition~\ref{thm:ppr_evidence_mass} states the resulting evidence-mass identity and its sufficient alignment condition.

Hyper-RAG enumerates higher-order paths with beam search~\cite{feng2026hyperrag}; Cog-RAG filters themes before entity retrieval~\cite{hu2026cog}; and IGMiRAG adjusts mining depth from estimated question complexity~\cite{hou2026igmirag}. ReMindRAG stores previous traversal experience in graph-edge embeddings~\cite{hu2026remindrag}. STITCH-RAG uses no path enumeration or query-time LLM call: STIBP scores the activated hypergraph structures, and PPR is restricted to the induced chunk graph. Table~\ref{q4} reports the combined retrieval latency, not the isolated cost of localization.

Corrective RAG, Adaptive-RAG, Self-RAG, and StructRAG decide whether retrieval or a particular structure is needed~\cite{yan2024corrective,jeong2024adaptive,asai2024selfrag,li2024structrag}. That routing decision is separate from the retrieval operator evaluated here. Long-context studies likewise show that evidence placement and concentration affect multi-hop reasoning even when the source fits within the context window~\cite{li2024retrieval,xu2024retrieval}. These results motivate selective retrieval but do not identify a preferred graph representation.

LinearRAG remains the closest end-to-end structural comparison. STITCH-RAG changes its two relevant design choices: local entity states replace the single name-matched entity representation, and continuous STIBP scores replace binary PPR seeds. STITCH-RAG exceeds LinearRAG in Contain-Acc by 15.5 points on HotpotQA and 6.7 on 2Wiki, and in Recall@8 by 7.5 and 5.7 points. These aggregate margins do not separate the effects of representation, initialization, and localization; the controlled replacements in Table~\ref{tab:ablation_extended} provide the component-level evidence available in this study.

The deployment regime also matters. RAG is preferable to knowledge internalization when the corpus is large or updated independently of model training, whereas fine-tuning can remain competitive on narrow static domains~\cite{ovadia2024finetuning,balaguer2024rag}. Structured retrieval is most useful for questions that require cross-passage evidence rather than direct lookup~\cite{xiang2025whentousegraphs}. STITCH-RAG further assumes that repeated queries amortize its one-time hypergraph construction cost.
\end{document}